%% file: main.tex
\documentclass[10pt,twocolumn]{IEEEtran}

\usepackage{epsfig}
\usepackage{algorithm,algorithmic, color}
\usepackage{multirow} 
\usepackage{amsfonts,amsmath,color,amssymb,amsxtra,graphicx,algorithm,algorithmic}
\usepackage{subfigure}

\newcommand{\fref}[1]{Figure~\ref{#1}}

\newcommand{\sref}[1]{Section~\ref{#1}}
\newcommand{\tref}[1]{Table~\ref{#1}}

\renewcommand{\and}{\hspace{0.5cm}}

\newtheorem{proposition}{Proposition}

\begin{document}

\graphicspath{{./Figures/}}

\title{Effective Delay Control in Online Network Coding}

\author{Jo\~ao Barros\IEEEauthorrefmark{1} \and Rui A. Costa\IEEEauthorrefmark{2} \and Daniele Munaretto\IEEEauthorrefmark{3} \and Joerg Widmer\IEEEauthorrefmark{3}\\[0.5cm]
\IEEEauthorrefmark{1}Instituto de Telecomunica\c{c}\~oes, Faculdade de Engenharia da Universidade do Porto, jbarros@fe.up.pt\\
\IEEEauthorrefmark{2}Instituto de Telecomunica\c{c}\~oes,
Faculdade de Ci\^encias da Universidade do Porto, ruicosta@dcc.fc.up.pt\\
\IEEEauthorrefmark{3}DoCoMo Euro-Labs,
Munich, Germany,
lastname@docomolab-euro.com
}

\maketitle
\begin{abstract}
Motivated by streaming applications with stringent delay constraints, we consider the design of online network coding algorithms with timely delivery guarantees. Assuming that the sender is providing the same data to multiple receivers over independent packet erasure channels, we focus on the case of perfect feedback and heterogeneous erasure probabilities. Based on a general analytical framework for evaluating the decoding delay, we show that existing ARQ schemes fail to ensure that receivers with weak channels are able to recover from packet losses within reasonable time. To overcome this problem, we re-define the encoding rules in order to break the chains of linear combinations that cannot be decoded after one of the packets is lost.  Our results show that sending uncoded packets at key times ensures that all the receivers are able to meet specific delay requirements with very high  probability.
\end{abstract}

\IEEEpeerreviewmaketitle

\input{intro}

\input{existing}
\input{delay}
\input{algo}
\input{results}
\input{imperfect}
\input{conclusions}

\bibliographystyle{IEEEtran}
\bibliography{main}

\end{document}

%% file: intro.tex
\section{Introduction}

The issue of delay between data transmission and successful delivery
to the receiving application is arguably one of the key concerns when
applying coding ideas to networking problems. This is particularly
true for network coding, where nodes combine multiple packets by means
of algebraic operations and perform computationally heavy Gaussian
elimination algorithms to recover the sent data. Although there is
growing consensus that both in wireless broadcast
scenarios~\cite{nc_xor,Fragouli2006Primer}
network coding can bring benefits in terms of throughput and
robustness, the fact that a receiver may have to wait for a
considerable number of packets, before it can decode the data,
justifies the question whether and how network coding can be used in
scenarios with stringent end-to-end delays.

In the seminal paper of Ahlswede, Cai, Li, and
Yeung~\cite{ahlswede2000nif}, which shows that network coding is
required to achieve the multicast capacity of a general network, the
problem is formulated in an information-theoretic setting, where delay
and complexity are not taken into account. Delay is also not a primary
concern of the algebraic framework in~\cite{koetter} and of the random
linear network coding method~\cite{Ho2003,practical_coding}, in which each
node in the network selects independently and randomly a set of
coefficients and uses them to form linear combinations of the data
symbols (or packets) it receives. When intermediate nodes cannot
perform coding operations and applications are able to tolerate some
delay, fountain codes (e.g.,~Raptor codes~\cite{shokrollahi2006rc})
emerge as a viable solution offering low coding overhead as well as
near-optimal throughput over packet erasure channels.

Clearly, in all of these instances, coding is performed in a
feedforward fashion. The encoders upstream are oblivious to packet
loss downstream and their coding decisions do not exploit any feedback
information.  In contrast, the property that transmitted packets are
linear combinations of subsets of packets available at the sender
buffer suggests that network coding protocols could be enhanced by
modifying the content of the acknowledgments typically provided by
transport protocols.
Instead of acknowledging specific packets, each
destination node of a unicast or multicast session can send back
requests for degrees of freedom that increase the dimension of its
vector space and allow for faster decoding.

Recent contributions that pursue this idea
(e.g.,~\cite{on_feedback_for_NC, Online_broadcasting_with_NC}) focus
mostly on end-to-end reliability with perfect feedback, i.e., complete
and immediate knowledge of the packets stored at each receiver. The
source node reacts by sending the most innovative linear combination
that is useful to most destination nodes. Throughput optimal network
coding protocols following this concept appear
in~\cite{ARQ,sundararajan2008onc}, which introduce the useful notion
of {\it seen} packet as an abstraction for the case in which a packet
cannot yet be decoded but can be safely removed from the sender
buffer. Removing packets in a timely fashion has obvious benefits in
terms of queue length. By using the feedback information to move a
{\it coding window} along the sender buffer instead of mixing fixed
sets of packets (also called generations~\cite{practical_coding}),
these protocols perform {\it online} network coding in the sense that
they adapt their coding decisions based on the erasure patterns
observed in the network.

Realizing that existing solutions do not yet cover the full range of trade-offs between throughput and delay, in particular when users experience different packet loss probabilities, we set out to provide end-to-end delay control for online network coding with feedback. Our main contributions are as follows:\\
$\bullet$ {\it Delay Analysis:} We provide an analytical framework to
  evaluate the delay performance of online network coding algorithms
  that leverage feedback for increased reliability. The novelty of our
  approach lies in observing how each erasure event affects the chains
  of undecoded linear combinations that build up at the receiver
  buffer. Moreover, we can map the information backlog between
  receivers to an appropriate random walk on a high dimensional
  lattice, which brings further insight into the delay behavior.\\
$\bullet$ {\it Online Network Coding Algorithms with Delay Constraints:}
  Using the knowledge of the chain length at each receiver, we
  identify simple ways of limiting the delay by means of informed
  encoding decisions. In particular, we show the benefits of sending
  uncoded packets to alleviate the delay of weaker receivers.

Our work differs from~\cite{ARQ} in that we consider heterogeneous users with different erasure probabilities and take the end-to-end delay to be our main figure of merit. Also centered around equal erasure probabilities for all users, the contribution in~\cite{sundararajan2008onc} focuses on the two user case and uses only the binary field, whereas, in contrast also with~\cite{Online_broadcasting_with_NC}, we consider also larger field sizes and larger number of users. A different method to limit the delay is to mix packets in such a way that at least some of the receivers are able to decode a symbol immediately upon receiving a new packet. If no feedback information is available, the best one can do is to choose the packets randomly and optimize only the number of packets to be combined~\cite{GC}, an approach which appears adequate for highly constrained scenarios such as data preservation in sensor networks. Results on the optimum degree distributions with respect to network dynamics and topology can be found in~\cite{Resilient_Coding}. The use of feedback under similar assumptions was explored in~\cite{costa:2008}. The main difference between~\cite{costa:2008} and this contribution is that here we provide analytical results, consider higher fields and do not enforce immediate decoding. We believe that our algorithms are able to reach a larger set of operating points in the delay-throughput plane and are thus well suited for streaming applications with stringent delay requirements, where network coding has already proved to yield competitive solutions (see, e.g.,~\cite{wang2004dca}).

The remainder of the paper is organized as
follows. \sref{sec:existing} introduces terminology and describes the
core ideas of online network coding with feedback. Our analytical
framework for evaluating the end-to-end delay is outlined in
\sref{sec:delay} with results on the relationship between erasure
patterns, undecodable chains, and incurred delays. \sref{sec:algo}
provides solutions for effective delay control and the corresponding
performance results are presented in \sref{sec:results}.  In
\sref{sec:imperfect}, we briefly discuss the implications of imperfect
feedback and conclude the paper in \sref{sec:conclusions}.

%% file: existing.tex
\section{Essential Background}
\label{sec:existing}

\subsection{System Setup}

Suppose that a single queue sender wants to transmit a stream of packets to multiple receivers. For simplicity, we assume that packets arrive at the sender in a certain order (older packets first) and are readily available at the sender for encoding and transmission. Each receiver $i$ is connected to the sender via a separate packet erasure channel, which takes one packet per time slot and loses a packet with probability $\epsilon_i$. Packets are lost independently across channels and time slots and receivers are able to detect when a packet is missing. Since the sender has access to perfect feedback (without errors, losses, or delay), it can make encoding decisions based on the buffer state of each receiver.

\subsection{ARQ for Network Coding}

The reference system for our analysis is the ARQ for network-coding (ANC) scheme presented in~\cite{ARQ}, which was shown to be throughput optimal for the case of Poisson arrivals, perfect feedback, and identical erasure probabilities on all channels. In this scheme, the sender transmits linear combinations of the packets in its queue, where the decision which packets to combine relies on the concept of \emph{seen packets}. A packet ${\bf p}$ is said to be seen by a receiver, if the receiver is able to construct a linear combination of the form ${\bf p}+{\bf q}$, such that ${\bf q}$ is a linear combination of packets that are newer than ${\bf p}$. In particular, a packet is seen when it can be decoded. The sender always transmits a packet that is a combination of the last (i.e., oldest) \emph{unseen} packets of each of the receivers. This ensures that the last unseen packet will now be seen by all receivers which receive the coded packet.

A packet can be dropped from the sender queue whenever it was
\emph{seen} (but not necessarily decoded) by all receivers. This has
the agreeable property that queue sizes at the sender are kept small,
since the sender can drop packets before they are decoded at all
receivers, without compromising reliability. The expected queue size
was shown to be $O((1-\epsilon)^{-1})$ \cite{ARQ}. The basic operation
of this scheme is illustrated in Table~\ref{tab:example}, which lists
the sequence of packet receptions (OK) and erasures (E) and shows the
corresponding coding decisions made by the sender for a two receiver
case. This example shall be discussed in more detail below and in the
next section.

The scheme in~\cite{ARQ} was extended in \cite{sundararajan2008onc} with the goal of reducing the decoding delay, specifically for the two receiver scenario. Here, packets that are unseen at a receiver because all combinations containing that packet were lost are requested by the receiver at a later stage. In the example of Table~\ref{tab:example}, this happens in time slot 7, where receiver 2 requests ${\bf p_6}$, instead of ${\bf p_5}$, resulting in the transmission of ${\bf p_6}\oplus{\bf p_7}$. Packet ${\bf p_5}$ is only requested in time slot 9, after receiver 2 decoded ${\bf p_6}$.

\begin{table}
\caption{Example of Online Network Coding with ARQ}
\vspace{-1em}
\label{tab:example}
\begin{center}
\begin{tabular}{c|c|c|c}
Time Slot &Sent Packet& Receiver 1 & Receiver 2\\\hline
1&${\bf p_1}$ &  OK &E\\
2&${\bf p_1}\oplus{\bf p_2}$  &OK &OK\\
3&${\bf p_2}\oplus{\bf p_3}$  &OK &OK\\
4&${\bf p_3}\oplus{\bf p_4}$  & OK & OK\\
5&${\bf p_4}\oplus{\bf p_5}$  & OK & E\\
6&${\bf p_4}\oplus{\bf p_6}$  &OK&OK\\
7&${\bf p_6}\oplus{\bf p_7}$  &E& OK\\
8&${\bf p_7}$  & OK&OK\\\hline
9&${\bf p_5}\oplus{\bf p_8}$  & OK&OK\\
10&${\bf p_8}\oplus{\bf p_9}$  & E&OK\\
11&${\bf p_9}$  & OK&E\\
12&${\bf p_9}\oplus{\bf p_{10}}$  & OK&OK\\
\end{tabular}
\end{center}
\vspace{-2em}
\end{table}

%% file: delay.tex
\section{Delay Analysis}
\label{sec:delay}

Before proceeding with the analysis, it is important to clarify the
notion of delay in the context of online network coding
algorithms. Once an information packet ${\bf p}$ arrives at the sender
queue it will typically go through five different stages: (1) ${\bf
  p}$ is mixed with other packets by means of coding operations, (2)
${\bf p}$ is transmitted to the receivers, (3) ${\bf p}$ is {\it seen}
by the receiver, (4) ${\bf p}$ is decoded, and (5) ${\bf p}$ is
delivered to the application.  In the following we shall focus on the
decoding delay, which is measured as the number of slots between the
first transmission of an encoding of the packet and successful
decoding at the receiver.  Clearly, this delay subsumes the time it
takes for a packet to be seen and the time it takes for a seen packet
to be decoded.

\subsection{Two Receivers}

We start with the case of two receivers and assume without loss of
generality that the sender restricts its transmissions to uncoded
packets or XORs of two packets~\cite{ARQ}. Since we assume that all
packets that are necessary for encoding are readily available at the
sender, the incurred decoding delay depends only on the rules enforced
by the online network coding algorithm and the erasure patterns of the
two channels. In each time slot we have one of the erasure events
listed in \tref{tab:list}, which occur with the given event
probabilities.
\begin{table}[h]
\vspace{-0.5em}
\caption{List of Possible Erasure Events}
\vspace{-3em}
\label{tab:list}
\begin{center}
\begin{tabular}{c|p{4cm}|c}
Event & Description & Probability \\\hline
$A$& No erasures &$P(A)=(1-\epsilon_1)(1-\epsilon_2)$\\
$B$& Receiver 1 gets the coded packet and receiver 2 observes an erasure&$P(B)=(1-\epsilon_1)\epsilon_2$\\
$C$& Receiver 1 observes an erasure and receiver 2 gets the coded packet&$P(C)=\epsilon_1(1-\epsilon_2)$\\
$D$& Both receivers observe an erasure&$P(D)=\epsilon_1\epsilon_2$
\end{tabular}
\end{center}
\vspace{-1em}
\end{table}

An erasure event causes a receiver buffer to build up a {\it chain} of
length $K$, which we define as a set of $K$ independent linear
combinations involving $L>K$ symbols, which cannot be decoded by the
receiver.  In the scenario shown in \tref{tab:example}, receiver
2 suffers from losses (denoted by E) in time slots 1 and 5,
whereas receiver 1 obtains everything error free (OK) except for the
data transmitted in slot 7. Each erasure sets a mark for a new chain
of undecodable linear combinations, such that each chain begins
immediately after its preceding chain has been solved. For example, up
to slot 4 receiver 2 built up the chain $\{{\bf p_1}\oplus{\bf p_2},
{\bf p_2}\oplus{\bf p_3},{\bf p_3}\oplus{\bf p_4}\}$. The erasure in
slot 5 sets a mark for a new chain, which will involve ${\bf p_5}$ by
necessity. However, before that chain begins, the first chain grows to
$\{{\bf p_1}\oplus{\bf p_2}, {\bf p_2}\oplus{\bf p_3},{\bf
  p_3}\oplus{\bf p_4},{\bf p_4}\oplus{\bf p_6}, {\bf p_6}\oplus{\bf
  p_7}\}$.  Since receiver 1 experiences an erasure in slot 7, the
encoding rule forces the sender to transmit packet ${\bf p_7}$ in
uncoded fashion, which in turn allows receiver 2 to break its first
chain and recover packets $\{{\bf p_1},{\bf p_2},{\bf p_3},{\bf
  p_4},{\bf p_6},{\bf p_7}\}$. The second chain begins immediately in
slot 9 with $\{{\bf p_5}\oplus{\bf p_8}\}$, because packet $\bf p_5$
was not seen by receiver 2 in slot 5.  Note that an erasure event at
the leading receiver 1 is not enough to allow receiver 2 to break the
current chain. If the following packet is lost at receiver 2, as shown
in the example with the loss of ${\bf p_2}$ in time slot 11, the chain
will simply continue to grow.

Clearly, the decoding delay is deeply influenced by the length of
chains such as these and by the sender's ability to break them in a
timely manner --- in spite of randomly occurring packet erasures

\subsubsection{First-Order Analysis}

Assuming that channels 1 and 2 have erasure probabilities $\epsilon_1$
and $\epsilon_2$, respectively, and that the sender follows the simple
ARQ rules outlined in \sref{sec:existing}, we can describe the
chain duration in a probabilistic fashion.
\begin{proposition}
\label{prop:chain1}
After an erasure of type $B$ that starts a chain at receiver 2, the
chain remains unbroken for a duration of $T$ slots with probability
\vspace{-.5em}
\begin{eqnarray}
P(T)&=&\epsilon_1 (1-\epsilon_2)^2\sum_{t_1=0}^{T-1}(\epsilon_1
\epsilon_2)^{t_1}\cdot \nonumber \\ && \hspace{-0.5cm} \cdot \hspace{-1cm} \sum_{t_2,t_3:2t_2+t_3=T-1-t_1}\hspace{-1cm}(\epsilon_1 \epsilon_2 (1-\epsilon_1)(1-\epsilon_2))^{t_2}\cdot (1-\epsilon_1)^{t_3}.
\label{eq:sum}
\end{eqnarray}
\end{proposition}

\begin{proof}
  We start by observing that events of type $D$ only increase the
  delay until the chain can be decoded but do not otherwise affect the
  recovery process or the length of the chain. Therefore, there is
  nothing to lose from ignoring events of type $D$ for now and taking
  their impact into account only at a later stage. If we only take
  into account events of type $A$, $B$ or $C$, a chain starting with
  an erasure is only broken after an erasure event of type $C$ (in
  which receiver 2 obtains a packet missed by receiver 1) immediately
  followed by an event of type $A$ or $C$, in which receiver 2 obtains
  the uncoded symbol that will ultimately allow it to decode the
  chain. While the chain is unbroken, any occurrence of event pairs
  that are not $CC$ or $CA$ will add to the duration of the chain. Any
  occurrence of $D$ at any slot (including between the first and
  second events of the pairs we considered previously) will further
  increase the chain duration.

  Since the channel erasures are independent from slot to slot, we can
  think of all the occurrences of $D$, none of which affects the breaking
  of the chain in any way, as a contiguous block in the first slots
  after the erasure. With this assumption in mind, notice that for the
  chain to be broken $T+1$ slots after the erasure, in slot $T$ we
  must observe $C$, since the only pairs that break a chain are $CA$
  and $CC$. Thus, after the first slots (where $D$ was observed) and
  up to and including slot $T-2$, if we observe a $C$, it must be
  followed by the event $B$. Regarding the remaining slots, we can
  have isolated events $A$ and $B$ (only the ones not preceded by $C$,
  because those are already taken into account). Therefore, letting
  $t_1$, $t_2$ and $t_3$ represent the number of occurrences of the
  events $D$, $CB$ and $A\cup B$, respectively, we have that
\begin{eqnarray}
P(T)\hspace{-0.4cm}&=&\hspace{-0.4cm} \sum_{t_1=0}^{T-1}P(D)^{t_1} \hspace{-0.8cm}\sum_{t_2,t_3:2t_2+t_3=T-1-t_1}\hspace{-1cm} P(CB)^{t_2} (P(A)+P(B))^{t_3} \cdot \nonumber \\ && \cdot (P(CA)+P(CC)) \label{eq:long}
\end{eqnarray}

Notice that, since erasures in different slots are independent of each other, we
have that $P(CA)=P(C)P(A)$, $P(CB)=P(C)P(B)$ and $P(CC)=P(C)^2$. Thus,
substituting $P(A)$, $P(B)$ and $P(D)$ by the expressions in
\tref{tab:list} in (\ref{eq:long}) and simplifying the resulting
terms we obtain Equation (\ref{eq:sum}).
\end{proof}
Likewise, we can compute the distribution for the chains built up at the other receiver.
\begin{proposition}
After an erasure of type $C$, a chain at receiver 1 remains unbroken for a duration of $T$ slots with probability
\vspace{-.5em}
\begin{eqnarray}
P(T)&=&\epsilon_2 (1-\epsilon_1)^2\sum_{t_1=0}^{T-1}(\epsilon_1
\epsilon_2)^{t_1}\cdot \nonumber \\ && \hspace{-0.5cm} \cdot \hspace{-1cm} \sum_{t_2,t_3:2t_2+t_3=T-1-t_1}\hspace{-1cm}(\epsilon_1 \epsilon_2 (1-\epsilon_1)(1-\epsilon_2))^{t_2}\cdot (1-\epsilon_2)^{t_3}.
\label{eq:sum2}
\end{eqnarray}
\end{proposition}

\begin{proof}
  The proof follows analogously to the previous proof by swapping
  events $B$ and $C$.
\end{proof}

\subsubsection{Higher-Order Analysis}

As mentioned earlier, if a second erasure of type $B$ occurs while
receiver 2 is still processing an unbroken chain, the result can be
viewed as a marker that signals a future new chain. This second chain
will begin immediately after the receiver recovered from the current
chain. From then on, the distribution of the number of slots required
to break it is also given by Proposition~\ref{prop:chain1}. The delay
with which receiver 2 sees (and later decodes) the packet missed in
the second erasure will naturally be dependent on the number of slots
that pass between the second erasure and the breaking of the first
chain. A similar argument applies to the $k$th erasure event that
takes place while receiver 1 is recovering from the first chain. The
result is a marker for a $k$th chain, whose overall delay is given by
$D=\sum_{i=0}^kT_i$,
where $T_0$ is the number of slots between the $k$th erasure and the
breaking of the first chain and $T_i$ denotes the time it takes to
break chain $i$. Notice that $T_0$ follows the same distribution as
$T_i$, because it counts the time slots between an erasure and the
breaking of a chain. Clearly, $D$ results from the sum of the
independent and identically distributed random variables
$T_0,T_1,\dots,T_k$, and for this reason is itself a random variable.

\subsection{Multiple Receivers}

As should be expected, determining the various forms of delay becomes
increasingly complex for larger numbers of receivers. To gain some
insight, we start by observing that, at any point in time there will
be at least one leader, where leader(s) at time slot $t$ are one or
more of the receivers which have received the maximum number of
packets up to time slot $t$. The following proposition describes an
important property of the leader status.
\begin{proposition}
\label{prop:leader_decodes}
  A receiver that became a leader at time $t$ and stays in the group
  of leaders until it receives one more packet at time $t+\delta$ is
  able to decode all packets included in any of the linear
  combinations transmitted until $t+\delta$.\footnote{A receiver may
    become part of the group of leaders when it receives a packet and
    the leaders do not. However, if the next packet or packets are lost
    and the receiver drops out of the group of leaders before
    receiving another packet, it will not be able to break its
    chain. This is analogous to the events $CB$ and $CD$ in the
    previous example.} A receiver continues to be able to decode
  immediately all coded packets at $t'>t+\delta$, provided it remains
  the leader or a member of the group of leaders.
\end{proposition}

\begin{proof} Assume the leaders lose a packet at time $t$ (otherwise
  no other receiver can become a leader) and suppose that they
  received the first $d$ packets up to that time. These packets carry
  encodings of at most the first $d+1$ information units. This ensures
  that leaders would have been able to decode a new packet, had they
  received the current transmission. A new receiver is now to become
  part of the leaders receives its next packet at time $t+\delta$ (and
  thus no other leaders receive a packet between $t$ and
  $t+\delta$). Since the coding algorithms are throughput optimal
  (i.e., each received packet is innovative) the receiver will have
  $d+1$ (coded) packets which are combinations of the first $d+1$
  original packets. It can thus solve the corresponding system of
  linear equations and decode all packets.
\end{proof}
As shown in the example of \tref{tab:example} for time slot 8, also
non-leaders may break chains and decode packets. However, as the
number of receivers increases and/or the erasure probabilities become
more heterogeneous, the probability that non-leaders can decode
becomes very small.

We use the fact that leaders can decode all packets to derive an upper
bound on the decoding delay in the multiple receiver case.  As we have
seen, the decoding delay is tightly connected to the time interval
between the moment in which a receiver ceases to be a leader and the
moment it is able to catch up and regain the leader status. Describing
the system in terms of the packets received by each of the receivers
leads to a state space which grows exponentially in the number of
receivers and is therefore intractable.

However, taking the point of view of one of the receivers, for
instance $R_1$, we can describe the evolution of the differences in
received packets between $R_1$ and the remaining receivers as a random
walk in an $(n-1)$-dimensional lattice, where $n$ is the number of
receivers. To develop some intuition, consider the case of three
receivers, denoted $R_1$, $R_2$, and $R_3$. Let $x_1$ denote the
difference of received packets between $R_1$ and $R_2$, and let $x_2$
describe the difference between $R_1$ and $R_3$. The state of the
system from the perspective of $R_1$ is thus described by the pair
$(x_1,x_2)$, which can be viewed as a point in two-dimensional
space. In each time slot, there are eight possible erasure events,
depending on whether each of the receivers suffers a packet loss or
not. If, in a given time slot, all receivers lose a packet or if there
are no packet losses, then the state does not change. In all other
cases, $x_1$ and $x_2$ will increase or decrease by one unit according
to the transition rules in \tref{tab:transition}, where once
again we use OK and E to denote successful reception and packet
erasure, respectively.

\begin{table}
\caption{Transition Rules for the Three-Receiver Case}
\vspace{-1em}
\label{tab:transition}
\begin{center}
\begin{tabular}{c c c c|l l}
Event & $R_1$ & $R_2$ & $R_3$ & Next State & Direction \\\hline
$E_0$&OK&OK&OK& $(x_1,x_2)$ &$\cdot$\\
$E_1$&OK&OK&E& $(x_1,x_2+1)$&$\uparrow$\\
$E_2$&OK&E&OK& $(x_1+1,x_2)$&$\rightarrow$\\
$E_3$&OK&E&E& $(x_1+1,x_2+1)$&$\nearrow$\\
$E_4$&E&OK&OK& $(x_1-1,x_2-1)$&$\swarrow$\\
$E_5$&E&OK&E& $(x_1-1,x_2)$&$\leftarrow$\\
$E_6$&E&E&OK& $(x_1,x_2-1)$&$\downarrow$\\
$E_7$&E&E&E& $(x_1,x_2)$&$\cdot$
\end{tabular}
\end{center}
\vspace{-1em}
\end{table}

Once we associate the erasure events with the corresponding
probabilities, which can be easily computed from the erasure
probabilities for each of the receivers, we obtain a random walk on a
two-dimensional lattice, as illustrated in
\fref{fig:random_walk}.

\begin{figure}
  \centering
 \includegraphics[width=8.5cm]{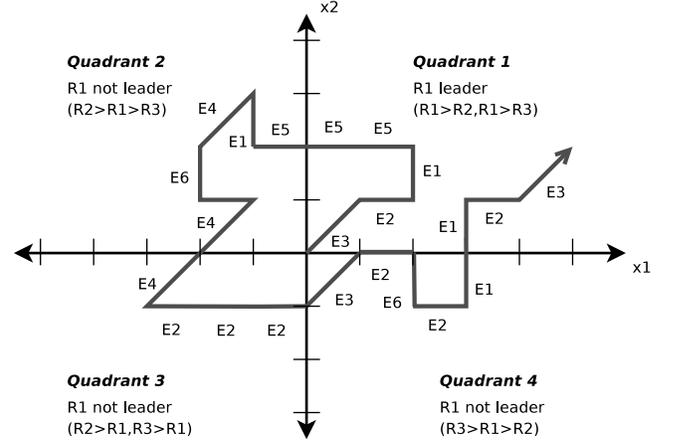}
 \vspace{-.5em}
  \caption{Random walk interpretation of the state evolution at receiver $R_1$.}
 \label{fig:random_walk}
 \vspace{-1em}
\end{figure}

Clearly, $R_1$ is a leader if and only if the coordinates $(x_1,x_2)$
lie in the first quadrant. In this case, $x_1$ and $x_2$ are both
positive (or zero) and $R_1$ has received either the same or a higher
number of packets than the other receivers. $R_1$ ceases to be a
leader, when its state position moves from the first quadrant to one
of the other three. Conversely, it becomes a leader again if its state
position moves back to quadrant one. The length of time $R_1$ spends
in each of the quadrants depends only on the erasure probabilities, or
equivalently the probabilities of erasure events $E_1\dots E_7$.

The proposed random walk model proves to be very useful for computing
upper bounds on the decoding delay experienced by receiver $R_1$. When
$(x_1,x_2)$ is in the first quadrant, we have that $R_1$ is the leader
and, thus, by Proposition~\ref{prop:leader_decodes}, every received
packet is immediately decoded, which is equivalent to zero
delay. Therefore, if $n$ is the number of slots in which $R_1$ is a
leader and $e$ is the number of erasures observed by $R_1$ during
those slots, $n-e$ packets have zero delay.

When $(x_1,x_2)$ lies outside the first quadrant (or, equivalently, when $R_1$ is not a leader), we can upper bound the delay of the packets transmitted in those slots as follows. Let $t_1$ be a time slot such that $(x_1,x_2)$ is in the first quadrant at slot in slot $t_1-1$ and elsewhere in slot $t_1$. Furthermore, let $t_2$ be the first time slot after $t_1$ such that  $(x_1,x_2)$ is again in the first quadrant in slot $t_2+1$. This means that during the slots between $t_1$ and $t_2$, $R_1$ is not a leader. In the worst case, by Proposition \ref{prop:leader_decodes}, all the packets transmitted between $t_1$ and $t_2$ will only be decoded at slot $t_2+1$. Hence, we can upper bound the delay of all these packets by $t_2-t_1$.

It is worth pointing out that generalizing this idea from three receivers to $n$ receivers forces
us to consider random walks in $n$-dimensional lattices. The class of random
walks we need can be deemed as untypical on several counts: (a) they assign non-uniform
probabilities to different directions by virtue of the properties of
online network coding, and (b) they admit the possibility that a node
stays in the same position. Close inspection of the related literature in 
probability theory reveals that the complete mathematical
characterization of integer random walks --- even for uniform
distributions in two dimensions --- offers non-trivial difficulties. A large
body of work is concerned with the number of points covered by the
random walk up to a certain time (see e.g.~\cite{caser:support}),
other contributions focus on hitting times on the coordinate
axis~\cite{cohen1992rwz} or among multiple random walks
~\cite{asselah}. At this time, providing a mathematical description of
the crossing times between quadrants of a non-uniform random walk is
clearly a daunting task, which justifies the use of numerical techniques
at the final stage of the proposed analysis.

Returning to the three-node case, Fig.~\ref{fig:multiple_chains}
illustrates the upper bounds on the decoding delay of receiver $R_1$
we obtain using the proposed methodology. As was to be expected, the
behavior of the cumulative distribution function can be very different
depending on whether the erasure probability of the receiver of
interest is equal to the erasure probabilities of the other receivers,
or lower, or higher.  With a higher erasure probability, as shown in
the bottom curve with $\epsilon_1=0.25$, $\epsilon_2 = 0.2$, and
$\epsilon_3 = 0.1$, almost all of the chains are broken only after the
leading receivers have received \emph{all} of their packets, resulting
in a very high decoding delay.  As an extreme case, if one of the
receivers has a perfect channel, there are no opportunities for any of
the other receivers to break chains.

\begin{figure}[t!]
  \centering
 \includegraphics[width=7cm]{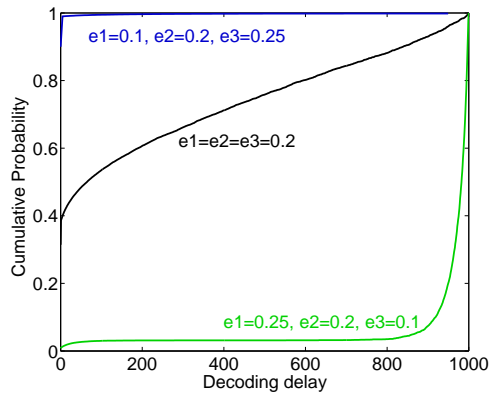}
 \vspace{-1em}
  \caption{CDF for the upper bound on the decoding delay of receiver
    $R_1$ for the
    proposed random walk model based on 10000 simulations with 1000
    slots for three different cases of erasure probabilities.}
\label{fig:multiple_chains}
 \vspace{-1em}
\end{figure}

Any effort to control the delay by means of informed coding decisions
amounts to creating opportunities for non-leading nodes to achieve the
leader status. In our random walk interpretation, this is equivalent
to pushing the target receiver to the first quadrant, whenever it has
spent more than an acceptable amount of time in the other three
quadrants. From a mathematical point of view, changing the encoding
rule corresponds to altering the probabilities assigned to each
direction of the random walk, thus making trajectories towards the
first quadrant more likely. Achieving this goal in practice is the
topic of the next section.

%% file: algo.tex
\section{Online Network Coding with Delay Constraints}
\label{sec:algo}

In the following, we present online network coding algorithms that are
targeted towards effective delay control. The underlying system setup
is the one we described in \sref{sec:existing}.


\subsection{Systematic Online Network Coding}

As argued in the previous section, the desirable properties of ANC in terms of
throughput optimality and a small sender queue size come at the expense of 
a potentially high decoding delay. We will now show that by allowing for 
some flexibility with respect to the sender's queue size, we can significantly 
reduce the average delay --- without sacrificing throughput.\\[-.5em]

{\bf Systematic Online Network Coding (SNC)}: \emph{A packet that is
  transmitted by the sender for the first time, is sent
  uncoded. Whenever the current leader suffers a packet loss, the
  next packet transmitted by the sender is a linear combination
  containing the last unseen packet of each receiver.}\\[-.5em]

An example of the SNC algorithm is given in \tref{tab:example2}. It
is not difficult to see that SNC does quite the opposite of ANC, in the sense that the
packets sent after receiver 2 loses a packet remain uncoded, whereas ANC
enforces their encoding. Conversely, a packet loss at receiver 1 causes
the transmission of the coded packet ${\bf p_1}\oplus{\bf p_7}$ in
time slot 8, whereas in the case of ANC the same event causes the 
transmissions of an uncoded packet.

The average queue size at the sender increases to
$\Omega((1-\epsilon)^{-2})$, the same as for traditional random linear
networking coding over all packets in the sender's queue \cite{ARQ},
since, for the algorithm to guarantee reliable communication, 
packets can only be removed from the sender's queue after they
have been successfully decoded at all receivers.

\begin{table}
\caption{Example of Systematic Online Network Coding}
\label{tab:example2}
\vspace{-0.5em}
\begin{center}
\begin{tabular}{c|c|c|c}
Time Slot &Sent Packet& Receiver 1 & Receiver 2\\\hline
1&${\bf p_1}$ &  OK &$E$\\
2&${\bf p_2}$  &OK &OK\\
3&${\bf p_3}$  &OK &OK\\
4&${\bf p_4}$  & OK & OK\\
5&${\bf p_5}$  & OK & $E$\\
6&${\bf p_6}$  &OK&OK\\
7&${\bf p_7}$  &$E$& OK\\
8&${\bf p_1}\oplus{\bf p_7}$  & OK&OK\\\hline
9&${\bf p_8}$  & OK&OK\\
10&${\bf p_9}$  & E&OK\\
11&${\bf p_5}\oplus{\bf p_9}$  & OK&E\\
12&${\bf p_{10}}$  & OK&OK\\
\end{tabular}
\end{center}
\vspace{-2em}
\end{table}

\begin{proposition}
  With SNC, each received packet is innovative and the algorithm is
  thus throughput optimal.
\end{proposition}
\begin{proof}
  A packet is sent uncoded only if it was never transmitted
  previously. Such a packet is clearly innovative for each receiver
  that obtains it. Coded packets are combinations of the oldest
  unseen packets of all receivers, as in ANC. Their reception thus
  causes a receiver to see the oldest unseen packet, which corresponds
  to an increase of the dimension of the information subspace
  available at that receiver. The proof follows analogously to the
  proof of throughput optimality for ANC in Theorem 3 of \cite{ARQ}.
\end{proof}

Since most of the packets are sent uncoded, the average packet delay
is much smaller than with ANC. For the same reason, the number of
required encoding and decoding operations is vastly reduced compared
to ANC. In short, the SNC algorithm builds up chains over missing
packets, not over all of the packets that follow a packet loss.

For real-time traffic, achieving throughput optimality is usually not
possible, since packets cease to be useful to the application if
they are delivered only after a certain deadline. In such a case, the
sender has to give up on those packets. In this case, for SNC, only
the missing packet is skipped, and, whenever possible, the sender continues to try to
repair other missing packets within their deadline. ANC, on the other hand,
loses the whole chain, which leads to a substantial reduction in throughput.

\subsection{Online Network Coding with a Delay Threshold}

To further reduce the worst-case delay for each receiver, we can trade
off some of the throughput for a substantial reduction in delay. As a
simple first measure to reduce delay, we can retransmit a packet in
uncoded form, in case a packet
deadline is in danger of being violated. \\[-.5em]

{\bf Systematic Online Network Coding with a Delay Threshold (SNCT)}:
\emph{Let $t$ be the current time slot and $t^i_j$ be the deadline for
  packet ${\bf p_j}$ at receiver $i$.  The sender proceeds as before,
  as long as $t<t^i_j-\delta$ for all undecoded packets $j$ and for
  each receiver $i$. In case a deadline for an undecoded packet $j$ is
  ``in danger'' and $t \geq t^i_j-\delta$, packet ${\bf p_j}$ is sent
  uncoded. In case receiver $i$ loses packet ${\bf p_j}$, this packet
  is repeated until the packet is either received or the deadline
  expires. In case multiple deadlines are in danger, the sender
  randomly picks a suitable undecoded packet and sends it in uncoded form.}\\[-.5em]

The same concept can be applied to ANC to obtain {\bf ANC with a Delay
  Threshold (ANCT)}. Again, the sender proceeds as with ANC, as long
as no deadlines for yet undecoded packets are in danger, and otherwise
transmits the corresponding packet in uncoded form. Once the packet is
repaired or the deadline expires, the sender resumes with transmitting
combinations of the last unseen packets of all receivers. Note that in
this case, it is no longer possible for an ANCT sender to discard
packets once they are seen by all senders. Consequently, the expected
queue size for ANC also increases to $\Omega((1-\epsilon)^{-2})$, the
same as SNC and SNCT.

%% file: results.tex
\section{Simulation Results}
\label{sec:results}

We investigate the performance of the different algorithms and the
impact of imposing decoding deadlines by means of simulation.  We use
a custom simulator with a full implementation of ARQ for network
coding without and with delay threshold (ANC, ANCT), as well as
systematic online network coding (SNC, SNCT). Packets are broadcast by
the sender and, as before, we use a simple channel model with
independent erasures for the different receivers. For all simulations,
the sender has 100 original packets to transmit to the receivers. The
metrics we consider are the mean and worst-case decoding delay, per
receiver throughput, and average and maximum queue size for the coding
at the sender, averaged over multiple simulation runs. The throughput
is calculated as the total number of packets decoded, divided by the
number of time slots it took until the last packet was decoded.

\begin{figure*}
\subfigure[Maximum and average decoding delay]{
  \label{fig:dec_delayN8}
  \includegraphics[width=0.5\textwidth]{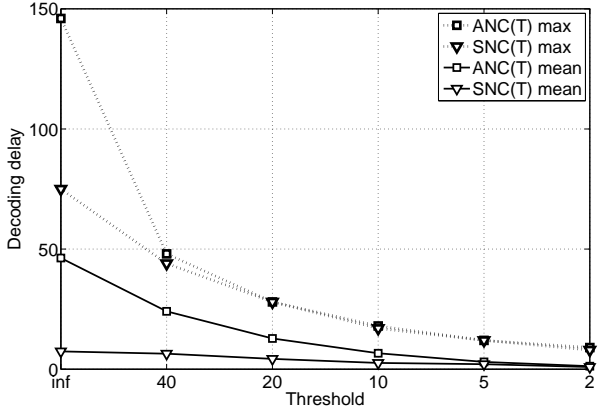}
  \vspace{-3.5em}
}
\subfigure[CDF of decoding delay]{
  \label{fig:sort_delayN8}
  \includegraphics[width=0.5\textwidth]{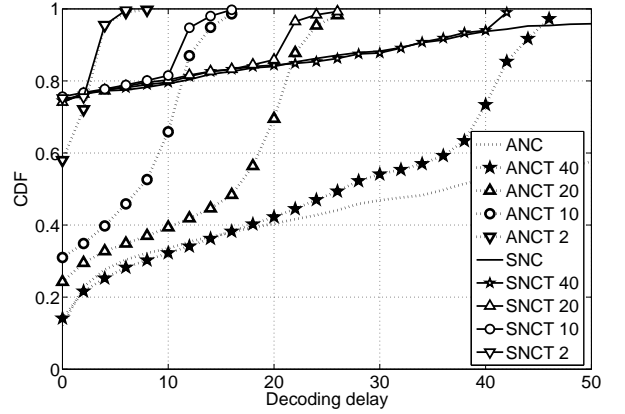}
  \vspace{-3.5em}
}
\subfigure[Average as well as minimum and maximum throughput]{
  \label{fig:thrN8}
  \includegraphics[width=0.5\textwidth]{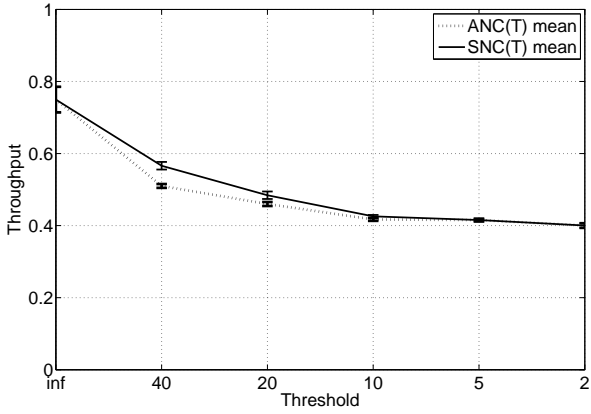}
  \vspace{-3.5em}
}
\subfigure[Maximum and average sender queue size]{
  \label{fig:queueN8}
  \includegraphics[width=0.5\textwidth]{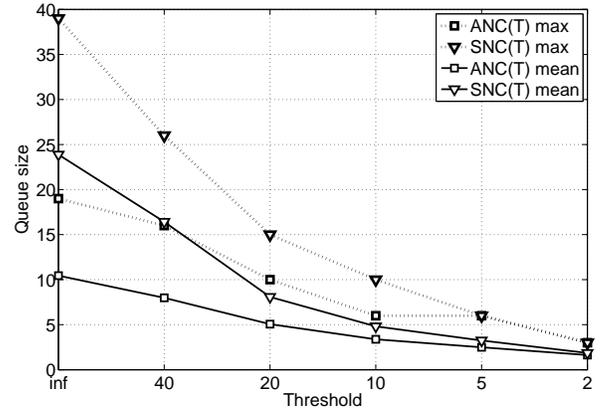}
  \vspace{-3.5em}
}
\caption{Analysis of different delay thresholds ($\infty, 40, 20, 10,
  5, 2$) for 8 receivers and $\epsilon=0.25$}
  \vspace{-1em}
\end{figure*}

\subsection{Impact of the Delay Threshold}

The value of the delay threshold is the main parameter that allows to
trade off throughput for delay. Low delay thresholds increase
fairness, since the difference in terms of number of received packets
between the leaders and other receivers is reduced. They also reduce
sender queue size for the same reason. For the simulations, we use a
receiver set size of $N=8$ and the channels to the receivers all have
erasure probability $\epsilon=0.25$. The sender continues to send
packets according to the respective algorithms, until all receivers
are able to decode all packets.

\fref{fig:dec_delayN8} shows the maximum and the mean decoding delay
for ANC and SNC without delay threshold (inf.) and for thresholds
between 2 and 40. Despite the same erasure probability for all
receivers, the worst case delay for ANC is close to duration of the
simulation, i.e., the worst receiver has a single chain which is not
broken until the very end. SNC fares much better with a worst case
delay only half as large. Since with 8 receivers, ANC transmissions
are usually combinations of 8 packets, a single chain may have many
packets missing (not just a single one as in the two receiver
case). While the oldest unseen packet that started the chain is seen
with ANC at the same time as it is seen with SNC, SNC can decode that
packet earlier than ANC is able to break the full chain.

The distribution of delay for all receivers is shown in
\fref{fig:sort_delayN8}. For SNC(T), 75\% of the packets are received
without delay. The higher the threshold, the later the CDF curve jumps
to a cumulative probability of 1, with the highest delay for plain
SNC, reaching 1 for a delay of ~75. ANC(T) has a much more varied
distributions of delays, with less than a third of the packets
received without delay in most cases.

Mean throughput is given in \fref{fig:thrN8}, and the error bars
indicate maximum and minimum throughput among the receivers. Optimum
throughput is at 0.75 for ANC and SNC without threshold. (Note that
the throughput of a receiver with erasure probability $\epsilon_i$
will be less than $1-\epsilon_i$ w.h.p.) The throughput reduction
caused by the threshold is on the order of 30\% to 40\%, depending on
the delay threshold chosen. As we decrease the threshold and enforce
lower delays, we can also see how maximum and minimum throughput get
closer and closer, leading to higher throughput fairness among the
receivers. Throughput is very close for both ANC and SNC; however, SNC
achieves a slightly more homogeneous distribution of throughput and
thus better fairness.

Also queue size decreases when a threshold is introduced, as shown in
\fref{fig:queueN8}. As expected, SNC requires a larger sender queue
size than ANC. However, the introduction of a threshold keeps the
maximum queue size to levels that are small enough to be easily
manageable at a sender, even for larger batches of packets to be
sent. Overall, the benefits in terms of delay improvements outweigh
the queue size increment for most realistic application scenarios.

\begin{figure*}
  \subfigure[Maximum delay (top) and average delay (bottom)]{
    \label{fig:delayN}
    \includegraphics[width=0.5\textwidth]{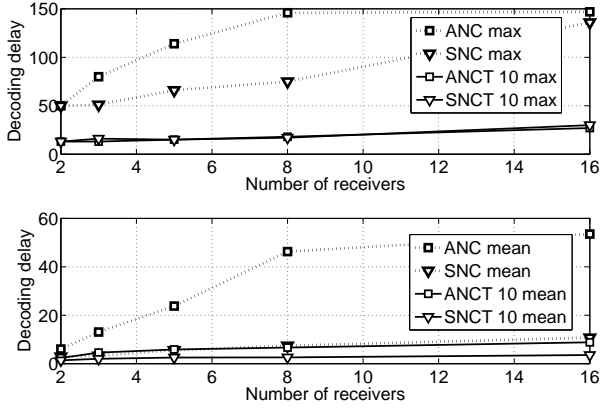}
    \vspace{-3.5em}
  }
  \subfigure[Average as well as minimum and maximum throughput]{
    \label{fig:thrN}
    \includegraphics[width=0.5\textwidth]{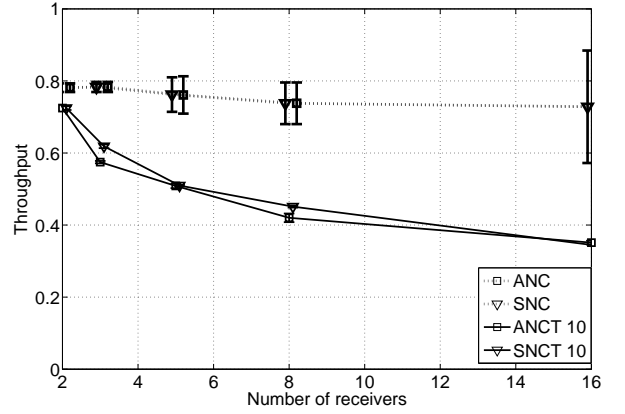}
    \vspace{-3.5em}
  }
  \caption{Analysis of different sizes of receiver sets from 2 to 16 for
    thresholds $\infty$ and 10, and $\epsilon=0.25$}
  \vspace{-1em}
\end{figure*}

\begin{figure*}
  \subfigure[Maximum delay (line) and average delay (bar)]{
    \label{fig:dec_eps}
    \includegraphics[width=0.5\textwidth]{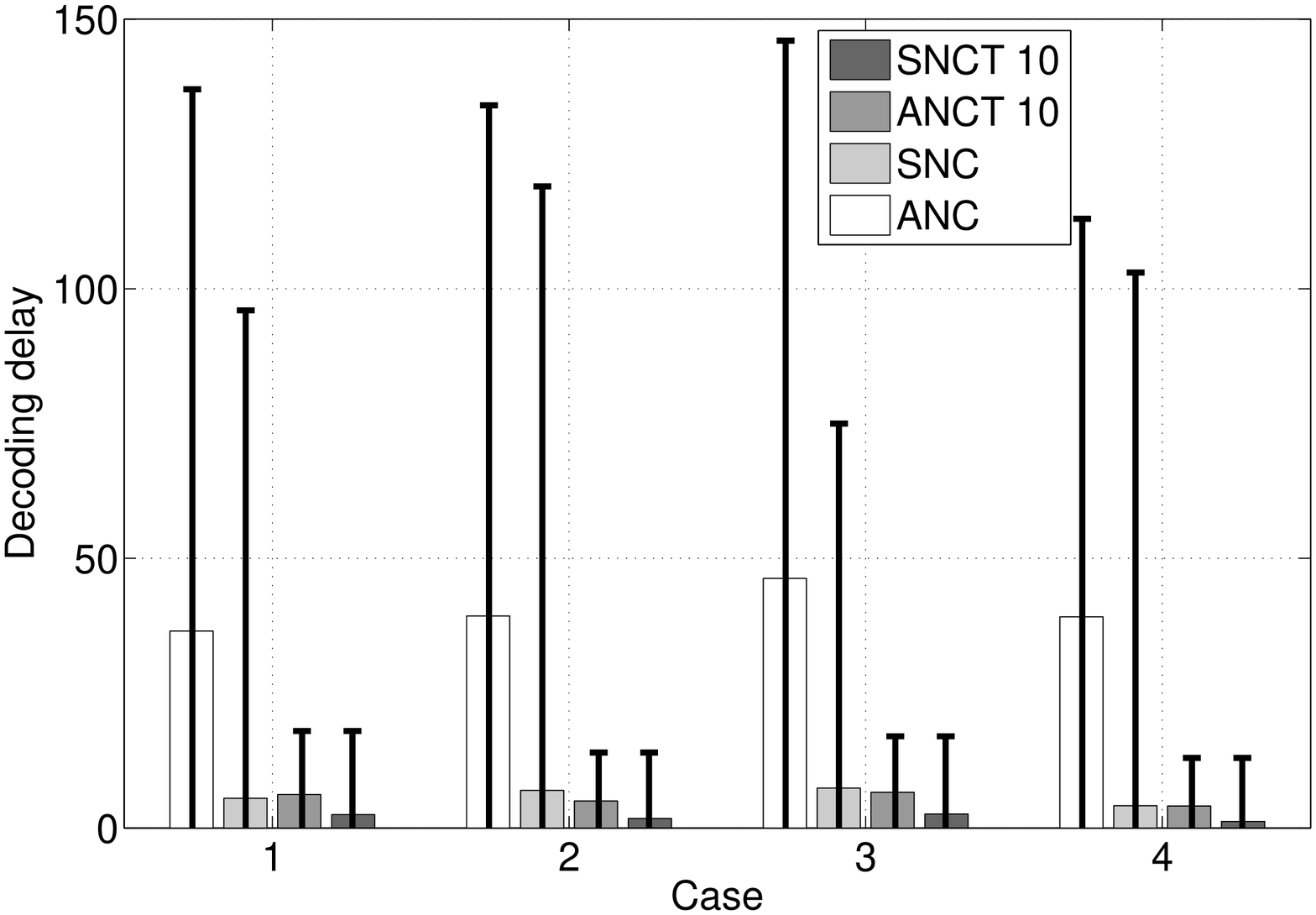}
    \vspace{-3.5em}
  }
  \subfigure[Average (bar) as well as minimum and maximum throughput (line)]{
    \label{fig:thr_eps}
    \includegraphics[width=0.5\textwidth]{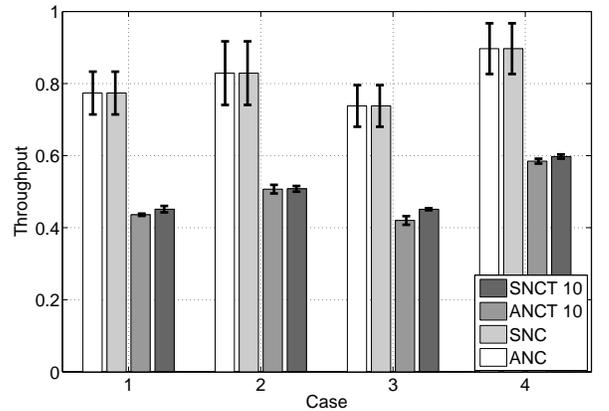}
    \vspace{-3.5em}
  }
  \caption{Analysis of four different cases of heterogeneous (1,2) and
    homogeneous (3,4) erasure probabilities for thresholds $\infty$
    and 10, and for 8 receivers}
  \vspace{-1em}
\end{figure*}

\subsection{Impact of the Size of the Receiver Set}

As the number of receivers increases, their heterogeneity in terms of
number of received packets will increase as well, even if they all
have the same channel erasure probability. Furthermore, as discussed
in \sref{sec:delay}, also the probability that non-leaders can decode
early becomes very small. The top graph in \fref{fig:delayN} shows
that the maximum decoding delay for ANC quickly approaches the
duration of the simulation of ~150 time slots. While SNC has a similar
worst case delay for very small and very large receiver sets, delays
for intermediate sizes are significantly lower. For two receivers,
chains with ANC are only missing a single packet which is repaired at
the same time, as the packet is repaired through a coded transmission
with SNC. Similarly, for large receiver sets, decoding the coded
repair packet with SNC becomes as hard as obtaining enough degrees of
freedom to decode a full chain.  Introducing a delay threshold clearly
lowers the worst case delay for both ANC and SNC. With a threshold of
10, packets are delivered at the latest after around 20 slots, due to
multiple packets that need to be repaired at the same time, and loss
of uncoded repair packets.

More importantly, for ANC the average delay even exceeds 50 time
slots. Given that packets originate over the course of the simulation,
the worst possible average delay is 75. This indicates that the vast
majority of packets for most of the receivers can only be decoded at
the end of the simulation and intermediate decoding is rare. Average
delay for SNC is comparable to the average delay for ANCT with a
threshold, while the average delay for SNCT is negligible.

As expected, throughput for ANC and SNC, as well as ANCT and SNCT is
comparable. However, the very short delay threshold of only 10 as used
in these simulations has quite a significant impact on throughput. For
algorithms with delay threshold, throughput may drop down to around
$\frac{1}{2}$ of that without threshold.


\subsection{Impact of the Erasure Probabilities}

We finally evaluate algorithm performance for four different
homogeneous and heterogeneous sets of erasure probabilities.  The
first case has 7 receivers with erasure probability 0.25, and a single
receiver with erasure probability 0.15. The second case has evenly
distributed erasure probabilities, with two receivers at erasure
probability 0.25, two receivers at 0.2, two receivers at 0.15, and two
receivers at 0.1. For the homogeneous cases, for all receivers we use
an erasure probability of 0.25 in case 3 and of 0.1 in case 4. As
before, there are 8 receivers and a delay threshold of 10 for ANCT and
SNCT.

Since for 8 receivers even the homogeneous case shows a worst case
delay close to the overall duration of the simulation, heterogeneous
erasure probabilities cannot exacerbate the performance, as can be
seen from \fref{fig:dec_eps}. For smaller sizes of receiver sets and
for higher numbers of transmitted packets, intermediate decoding
occurs more frequently and a much more significant advantage for the
use of delay thresholds in terms of maximum decoding delay can be
observed. Again, we see a throughput reduction caused by the use of
delay thresholds of around 40\%, as shown in
\fref{fig:thr_eps}. However, it is important to note that if packets
do become useless after their deadline expires and are discarded by
the receivers, the throughput reduction caused by deadline violations
is larger than that caused by the use of a delay threshold.

%% file: imperfect.tex
\section{Imperfect or Delayed Feedback}
\label{sec:imperfect}

So far, we assumed that the source always has perfect knowledge of the
decoding status of each receiver, which is usually not feasible in
practice. The overhead incurred by continuous feedback from all the
receivers may be too high. Furthermore, the feedback channel may
experience erasures, bit errors, and delay.

Consider the case where the feedback packet of a certain receiver is
lost. When performing the next coding decision, the source does not
know whether that receiver has received the previously sent packet or
not. In this situation, the source can 1) assume that the receiver
received that packet and now requires the next unseen packet, 2)
assume that the receiver missed the packet and still requires the
unseen packet reported previously, 3) perform a random experiment to
decide whether to consider the packet as received or lost, or 4)
ignore the receiver and not include any of its unseen packets in the
coding decisions until feedback from the receiver is heard again.

The tradeoff between throughput and delay is reflected also in the
treatment of missing feedback. In principle, the more optimistic the
sender is in its assumptions about packet reception, the higher the
expected throughput for the receivers, but the higher the risk of
increased chain lengths (and thus delay).  With the systematic
encoding algorithms, the main event that needs to be detected is
packet loss at the leader which allows to send coded repair
packets. Whenever this event is detected late due to feedback loss,
the same coded packet can be sent, but the delay of the repaired
packets increases by the corresponding amount. Whenever the event is
declared erroneously, throughput for the leader decreases since a
non-innovative packet is sent but other receivers may decode a
previously lost packet earlier.  These considerations are also
confirmed by preliminary simulations with our algorithms.  Thus, the
algorithms need continuous feedback from the leader, but they can
easily be modified to have less frequent feedback from other receivers
to reduce control overhead. Outdated information for the other
receivers is often unproblematic, since only the oldest unseen packet
is repaired. To this end, feedback about multiple packets can be
aggregated in a feedback vector \cite{nc_xor}.  Distributed algorithms
to effectively limit the amount of feedback have been developed, e.g.,
in the context of reliable multicast \cite{Fuhrmann2001scaling}.



%% file: conclusions.tex
\section{Conclusions}
\label{sec:conclusions}

Taking the decoding delay as our primary figure of merit, we analyzed
how the encoding rules of online network coding with feedback affect
its performance, in scenarios with homogeneous and heterogeneous
erasure probabilities. By describing the information backlog of
different receivers in terms of chains of undecoded packets and by
mapping their relative behavior in terms of a particular class of
random walks, we were able to show that without effective delay
control, weaker receivers are unlikely to recover from erasures within
reasonable time and that decoding success is essentially dependent on
having some advantage over other receivers in terms of received
packets.

This observation motivated us to re-design the encoding stage to
ensure that all receivers are able to decode at least some of the
packets most of the time. Surprisingly, as our extensive simulation
study shows, sending a large fraction of the packets in uncoded form
provides a throughput optimal solution with striking gains in terms of
decoding delay.  This is particularly useful for streaming
applications with stringent delay requirements, in particular when the
source codec is able to cope with missing packets and thus in-order
delivery of all packets is not really required.

As part of our ongoing work, we intend to combine our algorithms with
a state-of-the-art video codec and tune their performance to improve
the perceived video quality. This approach involves some
prioritization in the decision of which packets to combine, to take
into account their relative importance when decoding the video stream.
We also observe that sending a packet that is in danger of missing a
deadline in uncoded form is only a first step towards delay optimized
coding algorithms.  Coding strategies that ensure immediate decoding
of such packets at the lagging receiver, while providing the other
time-constrained receivers with innovative linear combinations, will
further improve performance. Finally, in order to design algorithms
that work well in practice, the implications of imperfect feedback
need to be investigated in more detail.